%% file: Information_Systems.tex
\documentclass[5p,times]{elsarticle}

\newtheorem{definition}{Definition}

\newtheorem{lemma}{Lemma}
\newtheorem{proof}{Proof}
\usepackage{booktabs} 
\usepackage[linesnumbered,ruled,vlined]{algorithm2e}
\usepackage{graphicx}
\usepackage{subfigure}
\usepackage{amsmath}
\usepackage{verbatim}
\usepackage{url}
\usepackage{multirow}
\usepackage{color}
\usepackage{amssymb}

\journal{Arxiv}

\begin{document}

\begin{frontmatter}



\title{Skyblocking for Entity Resolution}


\author{Jingyu~Shao,~
	Qing~Wang,~
	and~Yu~Lin
	}

\address{Research School of Computer Science, College of Engineering and Computer Science, The Australian National University, Canberra, ACT, 2600, Australia.}

\begin{abstract}
In this paper, for the first time, we introduce the concept of skyblocking, which aims to efficiently identify the "most preferred" blocking scheme in terms of a given set of selection criteria for entity resolution blocking. To capture all possible preferred blocking schemes, scheme skyline (i.e. blocking schemes on the skyline) has been studied in a multi-dimensional scheme space with dimensions corresponding to selection criteria for blocking (e.g. PC and PQ). However, applying traditional skyline techniques to learn scheme skylines is a non-trivial task. Due to the unique characteristics of blocking schemes, we face several challenges, such as: how to find a balanced number of match and non-match labels to effectively approximate a block scheme in a scheme space, and how to design efficient skyline algorithms to explore a scheme space for finding scheme skylines. To overcome these challenges, we propose a scheme skyline learning approach, which incorporates skyline techniques into an active learning process of scheme skylines. We have conducted experiments over four real-world datasets. The experimental results show that our approach is able to efficiently identify scheme skylines in a large scheme space only using a limited number of labels. Our approach also outperforms the state-of-the-art approaches for learning blocking schemes in several aspects, including: label efficiency, blocking quality and learning efficiency.

\end{abstract}


\end{frontmatter}


\input{Introduction}

\input{Problem_definition}

\input{Class_imbalance_learning}

\input{Scheme_skyline_learning_approach}

\input{Evaluation}

\input{Related_work}

\input{Conclusions}




\bibliographystyle{abbrv}
\bibliography{IS}
%
\input{Acknowledgment}
\end{document}

%% file: Introduction.tex
\section{Introduction}

Entity Resolution (ER) is of great importance in many applications, such as matching
product records from different online stores for customers, detecting
people's financial status for national security, or analyzing health conditions from different medical organizations \cite{christen2012data,fisher2015clustering}. 
Generally, entity resolution refers to the process of identifying records which represent the same real-world entity from one or more datasets \cite{wang2016clustering}.
Common ER approaches require the calculation of similarity between records in order to classify record pairs as matches or non-matches.
However, with the size of datasets growing, the similarity computation time for records increases quadratically. For example, given a dataset $D$, the total number of record pairs to be compared is $ \frac{|D|*(|D|-1)}{2} $ (i.e., each record is paired with all other records in $D$).
To reduce the number of compared record pairs, blocking techniques are widely applied, which can group potentially matched records into the same block \cite{wang2016semantic}. Since the comparison only occurs between records in the same block, blocking can reduce the number of compared record pairs to no more than $ \frac{m*(m-1)}{2} * |B| $, where $ m $ is the number of records in the largest block and $ |B| $ is the number of blocks.

In past years, a good number of techniques have been proposed for blocking \cite{draisbach2009comparison,wang2016semantic,fisher2015clustering,papadakis2014meta,papadakis2014supervised,shao2018active}. 
        Among these techniques, using blocking schemes is an efficient and declarative way to generate blocks. Intuitively, a blocking scheme takes records from a dataset as input, and groups the records using a logical combination of blocking predicates, where each blocking predicate specifies an attribute and its corresponding function. 
Learning a blocking scheme is thus the process of deciding which attributes, what their corresponding methods to compare values in attributes, and how different attributes and methods are logically combined so that desired blocks can be generated to satisfy the needs of an application.
Compared with blocking techniques that consider data at the instance level \cite{draisbach2009comparison,fisher2015clustering}, blocking schemes have several advantages: (1) They only require to decide what metadata, such as attributes and the corresponding methods, is needed, rather than data from individual records; (2) They provide a more human readable description for how attributes and methods are involved in blocking; and (3) They enable more natural and effective interaction for blocking across heterogeneous datasets. 

However, ER applications often require multi-criteria analysis of the varying preferences of blocking schemes. More specifically, given a scheme space that contains a large set of possible blocking schemes and a number of criteria for choosing blocking schemes, such as the commonly used ones: pair completeness (PC), pair quality (PQ) and reduction ratio (RR) \cite{christen2012survey}, how can we identify the \emph{most preferred} blocking scheme? 
Ideally, a good blocking scheme should yield blocks that minimize the number of record pairs to compare, while still preserving true matches as many as possible, i.e. optimizing all criteria
simultaneously. Unfortunately, the criteria for selecting blocking schemes are often competing with each other, for example, PC and PQ are negatively correlated in many applications, as well as RR and PQ \cite{christen2012survey}. In Fig.~\ref{fig:intro}, 
we can see that, a blocking scheme with an improved PC leads to a decrease in PQ, and conversely, a blocking scheme with an increase in PQ may have a decease in PC. Nonetheless, depending on specific applications, different blocking schemes can be considered as being "preferred" in different situations. For example, for a crime investigation where individuals need to be matched with a large databases of people, a high PC would be desired to make sure that potential criminals are included for investigation. On the other hand, in public health studies where each match corresponds to a patient with certain medical conditions, one may require a high PQ and want to be sure to only include patients that do have the medical condition under study \cite{christen2012data}.


	\begin{figure}
		
		\begin{minipage}[h]{0.21\textwidth} 
			\begin{tabular}{|c|c|c|}
				\hline
                Blocking &\multirow{2}{*}{PC} & \multirow{2}{*}{PQ}\\
				scheme & & \\
				\hline\hline
				$s_1$ & 0.13 & 0.76\\
				\hline
				$s_2$ & 0.31 & 0.99\\
				\hline
				$s_3$ & 0.58 & 0.76 \\
				\hline
				$s_4$ & 0.84 & 0.40 \\
				\hline
				$s_5$ & 0.86 & 0.50 \\
				\hline
				$...$ & $...$ & $...$ \\
				\hline
			\end{tabular}
		\end{minipage}
		\begin{minipage}[h]{0.2\textwidth} 
			\centering
			\includegraphics[height = 0.19 \textheight]{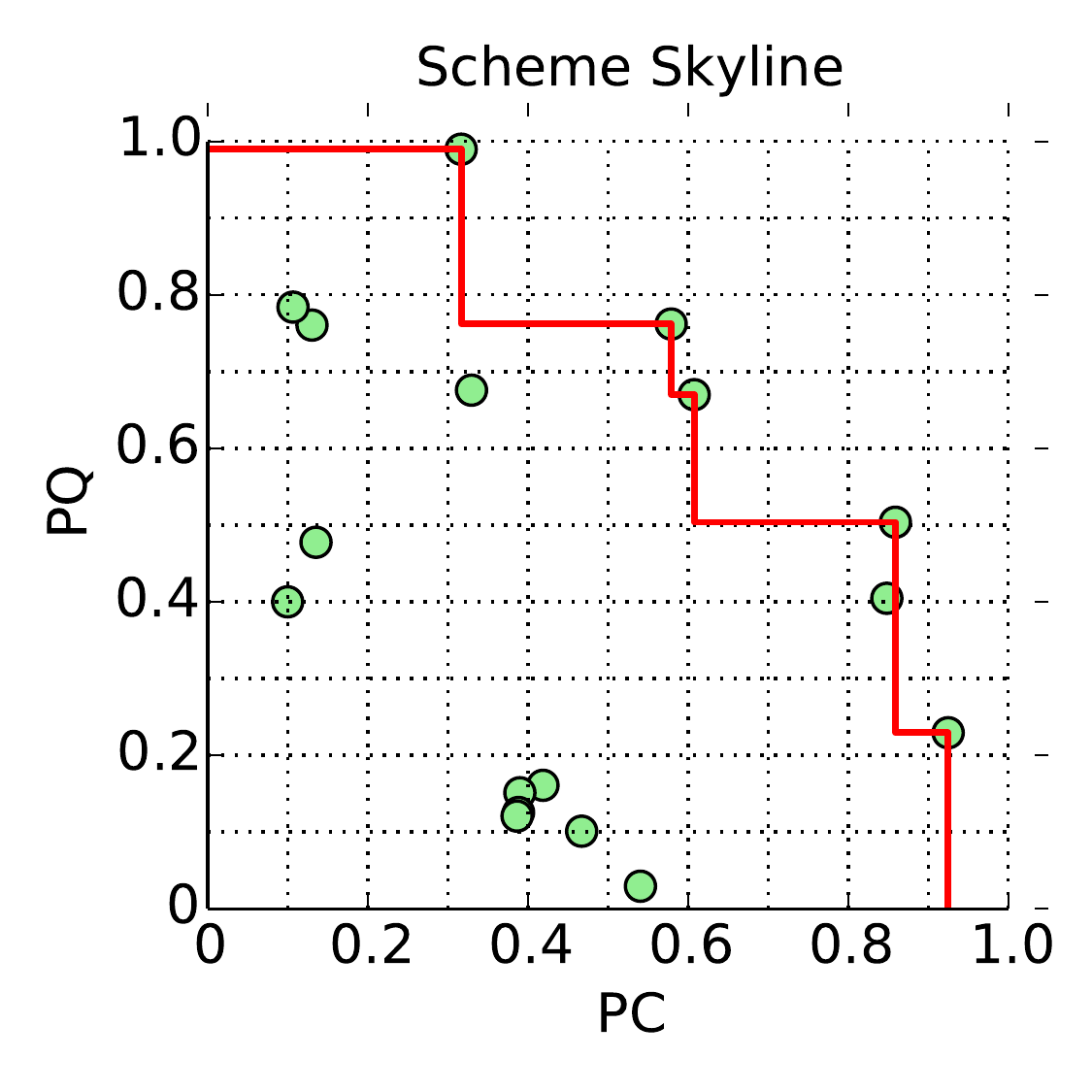}
		\end{minipage}
	\vspace{-2mm}
		\caption{An example for a scheme skyline (schemes in red line) from Cora Dataset, where blocking schemes (green points) are presented in a 2-dimensional space and their corresponding PC and PQ values are shown in the table.}
		\label{fig:intro}
		\vspace{-2mm}
	\end{figure}

Thus, to effectively learn a blocking scheme that is optimal in one or more criteria, previous work have specified conditions on the other negatively correlated criteria in the learning process \cite{kejriwal2015dnf}. For example, some approaches \cite{cao2011leveraging,michelson2006learning} aimed to learn a blocking scheme that can maximize both reduction ratio of record pairs with and without blocking, and coverage of matched pairs. Some others \cite{bilenko2006adaptive, kejriwal2013unsupervised} targeted to find a blocking scheme that can generate blocks containing all or nearly all matched record pairs with a minimal number of non-matched record pairs. Shao and Wang \cite{shao2018active} considered to select a blocking scheme with a minimal number of non-matched pairs constrained under a given coverage rate for matched pairs. However, setting such conditions perfectly is a challenging task, because conditions vary in different applications and it is \emph{unknown} which condition is appropriate for a specific application. If a condition is set too high, no blocking scheme can be learned; on the other hand, if a condition is set too low, the learned blocking scheme is not necessarily the most preferred one. Thus, the question that naturally arises is: can we efficiently identify all possible most preferred blocking schemes in terms of a given set of selection criteria for blocking?


To answer this question, we propose a scheme skyline learning approach for blocking, which incorporates skyline techniques into the learning process of blocking schemes on the skyline, i.e., \emph{scheme skyline}. The aim of learning scheme skyline is to find a range of blocking schemes under different user preferences, and each of such blocking schemes is the optimal with respect to one or more selection criteria. As shown in Fig.~\ref{fig:intro}, possible blocking schemes are marked as green points, and \emph{preferred} blocking schemes correspond to the points on the skyline (i.e., red line). This can thus provide a user with a direct view on all possible preferred blocking schemes for a given dataset and help make decisions on choosing the most preferred blocking scheme from a user's perspective.

\medskip
\noindent\textbf{Challenges } 
To the best of our knowledge, no scheme skyline learning approach has been previously proposed for entity resolution. Traditionally, skyline techniques have been widely used for database queries \cite{kalyvas2017survey, chomicki2013skyline,sarma2011representative}. The result of a skyline query consists of multi-dimensional points for which there is no other point having better values in all the dimensions \cite{chester2015explanations}. However, applying skyline queries technique to learning scheme skyline is a non-trivial task. Different from skyline queries in the database context \cite{chomicki2013skyline}, in which the dimensions of a skyline correspond to attributes and points correspond to records, scheme skylines in our work have the dimensions corresponding to selection criteria for blocking, and points corresponding to blocking schemes in this multi-dimensional space, called \emph{scheme space}. Learning scheme skylines becomes challenging, due to the following unique characteristics of blocking schemes. Firstly, it is hard to obtain match and non-match labels in real-world ER applications. Furthermore, it is well-known that match and non-match labels for entity resolution are often highly imbalanced, called \emph{the class imbalance problem} \cite{fisher2016active,christen2015efficient,wang2015efficient}. For example, given a dataset with two tables, each table containing 1,000 non-duplicate records, the total number of record pairs will be 1,000,000, but the number of true matches is no more than 1,000. The class imbalance ratio of this dataset is thus at most 1:1,000, which indicates that the probability of randomly selecting a matched pair is 0.001. However, to find scheme skylines effectively, we need a balanced training set with both match and non-match labels. Hence, the first challenge we must address is how to find a balanced number of match and non-match labels to effectively approximate a block scheme in a scheme space. Secondly, a scheme space can be very large, making an exhaustive search for blocking schemes on the skyline computationally expensive. As will be discussed in Section \ref{subsec-complexity}, if a blocking scheme is composed of at most $n$ different blocking predicates, the number of all possible blocking schemes can be $O (2^{n \choose [n/2]} )$ asymptotically. Nevertheless, only a small portion of these blocking schemes are on the skyline. Thus, the second challenge we need to act on is how to 
design efficient skyline algorithms to explore a large scheme space efficiently for finding blocking schemes on the skyline. Enumerating through the entire scheme space by checking a blocking scheme each time is obviously not an efficient solution in practice and infeasible for large datasets with hundreds or thousands of attributes.

\smallskip
\noindent\textbf{Contributions } 
To overcome the above challenges, in our approach, we hierarchically learn a scheme skyline based on blocking predicates to avoid enumerating all possible blocking schemes. A balanced active sampling strategy is also developed, which aims to select informative samples (i.e., a more balanced set of match and non-match samples) within a limited budget of labels. In summary, the contributions of this paper are as follows.

\begin{itemize}
	
	\item We formalize the scheme skyline learning problem and propose three novel algorithms for learning a scheme skyline. These enable us to efficiently learn the scheme skyline with significantly less label cost and help users to select the blocking scheme they need in terms of different constraints.
	
	\item To efficiently use limited labels, we tackle the class imbalance problem by converting it into the balance sampling problem. Then we propose the active sampling strategy to solve the balance sampling problem by actively selecting representative samples. 
    \item Three scheme skyline learning algorithms are proposed based on the active sampling and scheme extension strategies. The scheme extension strategy aims to reduce the search space and label cost by only considering potentially dominating schemes.
	
	\item We have evaluated our approach over four real-world datasets. Our experimental results show that our approach outperforms the state-of-the-art approaches in several aspects, including: label efficiency, blocking quality and learning efficiency. 

\end{itemize}

\medskip
\noindent\textbf{Outline } The rest of the paper is structured as follows. Section \ref{sec_pd} introduces the notations used in the paper and the problem definition. Section \ref{sec_cil} discusses our active sampling stretegy to the class imbalance problem. After that, three skyline algorithms for learning scheme skylines are presented in Section \ref{sec_ssla}, and the hierarchical scheme skyline learning
strategy has been discussed. Section \ref{sec_eva} discusses our experimental results, and Section \ref{sec_rw} introduces the related work. We conclude the paper in Section \ref{sec_con}.

%% file: Problem_definition.tex
\section{Problem Definition}
\label{sec_pd}

Let $D$ be a dataset consisting of a set of records, and each record $r_i \in D$ be associated with a set of attributes $A = \{a_1, a_2, ...,\\ a_{|A|} \}$. We use $r_i.a_k$ to refer to the value of attribute $a_k$ in a record $r_i$. Each attribute $a_k \in A$ has a domain $Dom(a_k)$. 
A \emph{blocking function} $h_{a_k} : Dom (a_k) \rightarrow U$ takes an attribute value $r_i.a_k$ from $Dom (a_k)$ as input and returns a value in $U$ as output.
A \emph{blocking predicate} $\langle a_k, h_{a_k} \rangle $ is a pair of attribute $a_k$ and blocking function $h_{a_k}$. 
Given a pair of records $\langle r_i, r_j \rangle$, a blocking predicate $\langle a_k, h_{a_k} \rangle$ returns \emph{true} if $h_{a_k}(r_i.a_k) = h_{a_k}(r_j.a_k)$ holds; otherwise \emph{false}.
For example, we may have \emph{soundex} as a blocking function for attribute \emph{author}, and accordingly, a blocking predicate $\langle \textit{author, soundex} \rangle$. For two records with values ``Gale'' and ``Gaile'' of attribute \emph{author}, $\langle \textit{author, soundex} \rangle$ returns \emph{true} because of $soundex(Gale) = soundex(Gaile) = G4$.

A \emph{(blocking) scheme} $s$ is a disjunction of conjunctions of blocking predicates (i.e. in the disjunctive normal form). For example, we may have $(\langle \textit{author, soundex} \rangle \wedge \langle \textit{title, exactmatch} \rangle)\vee (\langle \textit{author, soundex} \rangle \wedge \langle \textit{venue, exactmatch} \rangle)$ as a blocking scheme.
A blocking scheme is called a \emph{n-ary blocking scheme} if it contains $n$ distinct blocking predicates.
For example, a blocking scheme $(p_1 \wedge p_2) \vee (p_1 \wedge p_3)$ is a 3-ary blocking scheme, because it contains three blocking predicates, i.e. $p_1$, $p_2$ and $p_3$.

A \emph{training set} $T = (X, Y)$ consists of a set of feature vectors $X = \{x_1, x_2, ...,$ $x_{|T|}\}$ and a set of labels $Y = \{y_1, y_2, ..., y_{|T|}\}$, where each $y_i\in \{ M, N\}$ is the label (i.e. match or non-match) of the feature vector $x_i$ $(i=1, \dots, |T|)$. 
Given a pair of records $\langle r_i, r_j \rangle$, and a set of blocking predicates $P$, a \emph{feature vector} of $r_i$ and $r_j$ is a tuple $\langle v_1, v_2, ... v_{|P|} \rangle$, where each $v_k \ (k=1,\dots,|P|) $ is an equality value of either 1 or 0, describing whether the corresponding blocking predicate $\langle a_k, h_{a_k} \rangle \in P$ returns true or false.
For each feature vector $x_i$, a \emph{human oracle} $\zeta$ is used to provide a label $y_i\in Y$. 
If $y_i=M$, it indicates that the given pair of records refers to the same entity (i.e. \emph{a match}), and analogously, $y_i=N$ indicates that the given pair of records refers to two different entities (i.e. \emph{a non-match}). 
The human oracle $\zeta$ is associated with a budget limit \emph{budget$(\zeta) \geq 0 $}, which indicates the total number of labels the human oracle can provide.

Given a blocking scheme $s=s_1\vee s_2...\vee s_n$, where each $s_i$ ($i=1, \dots, n$) is a conjunction of blocking predicates, we can generate a set of pairwise disjoint blocks $B_s =\{b_1, \dots, b_{|B_s|}\}$, where $b_k\subseteq D$ ($k=1,\dots, |B_s|$), $\bigcup_{1\leq k\leq |B_s|} b_k = D$ and $\bigwedge_{1\leq i\neq j\leq |B_s|} b_i\cap b_j=\emptyset$. 
For any two records $r_i$ and $r_j$ in the dataset $D$, $r_i$ and $r_j$ are placed into the same block iff there exists a conjunction of block predicates $s_i$ such that, for each blocking predicate $\langle a_k, h_{a_k} \rangle$ in $s_i$, we have that $h_{a_k}(r_i.a_k) = h_{a_k}(r_j.a_k)$ holds.

Now we characterize blocking schemes under different perspectives using the notion of scheme skyline. 
Given a dataset $D$ and a set of blocking schemes $S$, a \emph{scheme space} over $D$ is a $d$-dimensional space consisting of points in $[0,1]^d$. Each blocking scheme $s\in S$ can be mapped to a point $p(s)$ in this scheme space such that $p(s) = (p_1.s, p\dots, p_d.s)$, where each $p_i.s$ $i\in [1,d]$ indicates the value of $s$ in the i-th dimension of this scheme space.

\begin{definition}
(Dominance) Given two blocking schemes $s_1$ and $s_2$, we say $s_1$ \emph{dominates} $s_2$, denoted as $s_1 \succ s_2$, iff $ \forall i \in d, p_i.s_1 \geq p_i.s_2$ and $\exists j \in d, p_j.s_1 > p_j.s_2$. 
\end{definition}

Based on the notion of dominance between two blocking schemes, we now define the notion of scheme skyline.

\begin{definition}
(Scheme skyline) Given a dataset $D$ and a set of blocking schemes $S$, a \emph{scheme skyline} $S^*$ over $S$ is a subset of $S$, where each scheme $s \in S^*$ is not dominated by any other blocking scheme in $S$, i.e. $S^* = \{s \in S: \nexists s' \in (S - S^*), s' \succ s \}$. 
\end{definition}

Without loss of generality, we will discuss a scheme space with $d=2$ in this paper, which has two dimensions: one indicates the PC of a blocking scheme and the other indicates the PQ. Both PC and PQ are commonly used measures for blocking \cite{christen2012survey}. For a pair of records that are placed into the same block, we call it a \emph{true positive}
if it refers to a match; otherwise, it is a \emph{false positive}. Similarly, a pair of records is called a \emph{false negative} if it refers to a match but the records are placed into two different blocks. For convenience, we use $tp(s)$, $fp(s)$ and $fn(s)$ to denote the numbers of true positives, false positives and false negatives in the blocks generated by $s$, respectively. The \emph{pair completeness} (PC) of a blocking scheme $s$ is the number of true positives $tp(s)$ divided by the total number of true matches, i.e. $tp(s)+fn(s)$, which measures the rate of true matches remained in blocks. We have:
  \begin{equation}
    PC(s) = \frac{tp(s)}{tp(s) + fn(s)}
  \end{equation}
The \emph{pair quality} (PQ) of a blocking scheme $s$ is the number of true positives $tp(s)$ divided by the total number of record pairs that are placed into the same blocks, i.e. $tp(s)+fp(s)$, which measures the rate of true positives in blocks. We have:
  \begin{equation}
    PQ(s) = \frac{tp(s)}{tp(s) + fp(s)}
  \end{equation}
Note that, it is possible and straightforward to extend the dimensions of a scheme space by taking into account other measures for blocking schemes, such as reduction ratio (RR) and execution time (ET) \cite{kopcke2010frameworks}.

We now define the problem of scheme skyline learning in a scheme space, which allows us to select desired blocking schemes depending upon different blocking criteria.

\begin{definition}
(Scheme skyline learning problem) Let $\zeta$ be a human oracle, $D$ be a dataset, $S$ be a set of blocking scheme over $D$, and $\Sigma$ be a d-dimensional scheme space. Then the \emph{scheme skyline learning problem} is to learn the scheme skyline $S^*$ over $S$ through actively selecting a training set $T$ from $D$, such that $|T| \leq budget(\zeta)$ holds.
\end{definition}

%% file: Class_imbalance_learning.tex
\section{Class Imbalance Learning}
\label{sec_cil}

\label{sec_as}
\begin{figure}
	\begin{center}
		\includegraphics[height=0.17\textheight]{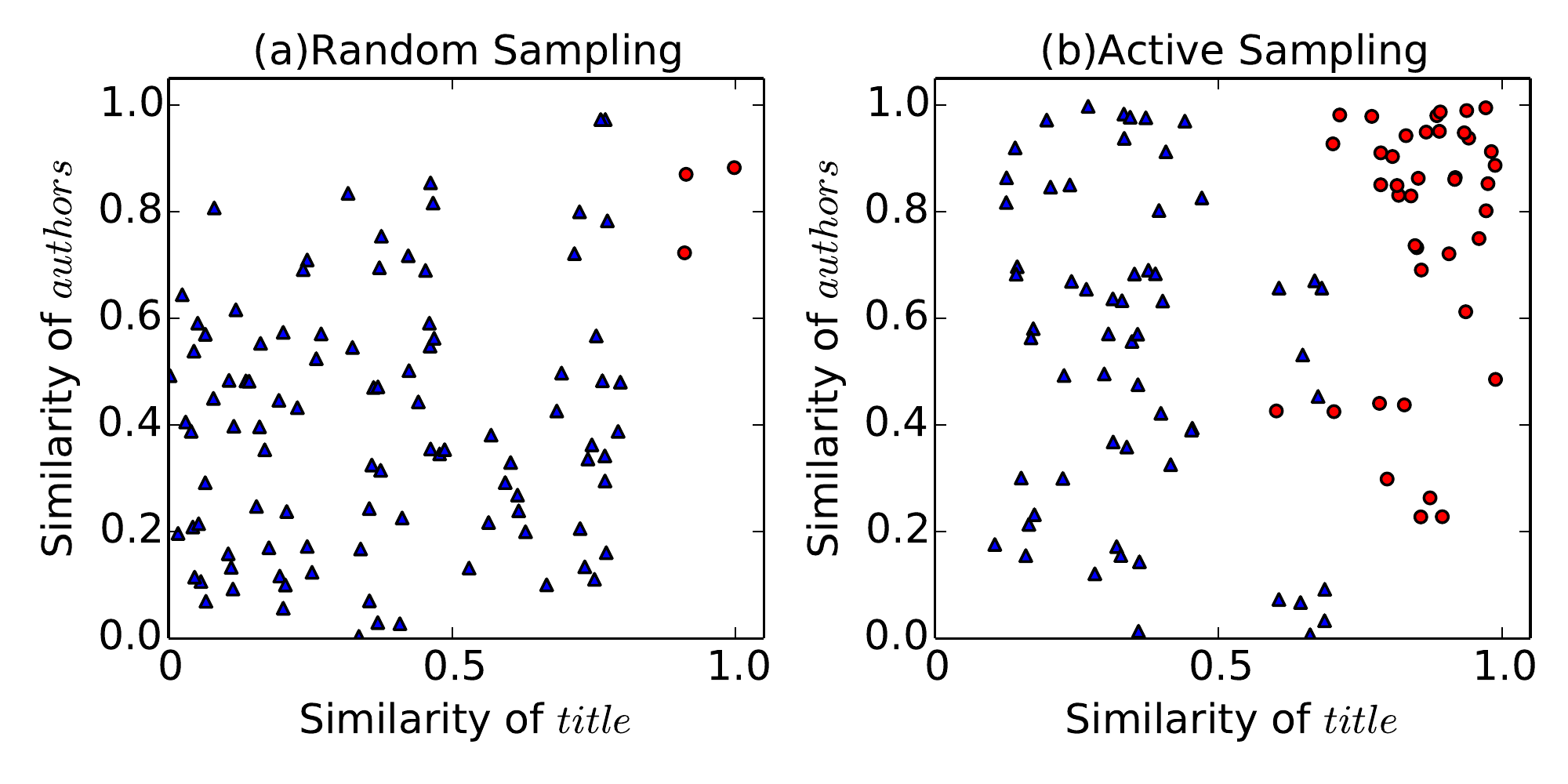}
	\end{center}
\vspace{-5mm}
	\caption{A comparison on the sample distribution of 100 samples from Cora dataset: (a) random sampling, and (b) active sampling (based on similarity of two attributes \emph{title} and \emph{authors}, where a red circle indicates a match sample and a blue triangle indicates a non-match sample.}
	\label{fig:3d}
	\vspace{-2mm}	
\end{figure}

In this section, we propose an active sampling strategy to tackle the class imbalance problem. This strategy aims to generate an informative training set with a small amount of label usage.

The class imbalance problem has been known to impede the performance of learning approaches in entity resolution \cite{bilenko2006adaptive,cao2011leveraging,kejriwal2013unsupervised}. Although the ratio of non-matches and matches may vary in different datasets, using random sampling can always lead to much more non-matches than matches in a training set, as shown in Fig.~\ref{fig:3d}.(a). Thus, supervised learning approaches that are based on random sampling need a large number of labels in order to gurantee the learning performance.

Previously, it has been reported \cite{arasu2010active,bellare2012active} that the more similar two records are, the higher probability they can be a match. To tackle the class imbalance problem, we observe that, by virtue of this correlation between similarity and labels, a balanced set of similar and dissimilar records likely represents a training set with balanced matches and non-matches.
Hence, we define the notion of balance rate.

\begin{definition}
	(Balance rate:) Let $s$ be a blocking scheme and $X$ a non-empty feature vector set, the \emph{balance rate} of $X$ in terms of $s$, denoted as $\gamma (s, X)$, is defined as:
	\begin{equation}
	\frac{|\{ x \in X | s(x) = true\}| - |\{ x \in X | s(x) = false\}|}{|X|}
	\end{equation}
\end{definition}

Conceptually, the balance rate describes how balance or imbalance of $X$ by comparing the number of similar samples to that of dissimilar samples in terms of a given blocking scheme $s$. The range of balance rate is $[-1, 1]$. If $\gamma (s, X) = 1$, there are all similar samples in $X$ with regard to $s$, whereas $\gamma (s, X) = -1$ means all the samples in $X$ are dissimilar pairs. In these two cases, X is highly imbalanced. If $\gamma (s, X) = 0$, there is an equal number of similar and dissimilar samples, indicating that X is balanced.

Then, based on the notion of balance rate, we convert the class imbalance problem into the balanced sampling problem, which is defined as follows:

\begin{definition}
	(Balanced sampling problem) Given a set of blocking scheme $S$ and a label budget $n \leq budget(\zeta)$, the \emph{balanced sampling problem} is to select a training set $T = (X, Y)$, where $|X| = n$, in order to:
	
	\begin{equation}
	\textbf{minimize }\sum_{s_i \in S} {\gamma(s_i, X)}^2
	\end{equation}
\end{definition}

For two different blocking schemes $s_1$ and $s_2$, they may have different balance rates over the same feature vector set $X$, i.e. $\gamma (s_1, X) \neq \gamma (s_2, X)$ is possible. The goal here is to find a training set that minimizes the balance rates in terms of the given set of blocking schemes. The optimal case is $\gamma(s_i, X)=0$, $\forall s_i\in S$ according to the objective function above, but sometimes it is impossible to achieve this in real world applications.
For example, a dataset may contain all the records from the same county, which would indicate all samples are similar samples for the blocking scheme $\langle \textit{county, exact-match} \rangle$.

In Fig.~\ref{fig:3d}, we compare samples identified using our active sampling strategy with the ones found using random sampling. Due to the class imbalance problem, the samples from random sampling are almost all non-matches, whereas the samples from active sampling contain much more matches than the samples from random sampling.

%% file: Scheme_skyline_learning_approach.tex
\section{Scheme Skyline Learning Approaches}\label{sec_ssla}

Now we propose three algorithms for scheme skyline learning. Our algorithms are designed upon the scheme extension strategy, which begins with blocking schemes containing only one blocking predicate, then extends with other blocking schemes in either conjunction or disjunction form.


\subsection{Scheme Extension Strategy}
In general, given a blocking scheme $s$, there are two kinds of extensions, through which we can extend $s$ to another blocking scheme: conjunction and disjunction. Hence if a conjunction or disjunction can be decided before extension, the searching space will be reduced by half.

Let $s_1$ and $s_2$ be two blocking schemes. We have the monotonicity property of PC, in terms of either the disjunction of $s_1$ and $s_2$ or the conjunction of $s_1$ and $s_2$, based on the following lemma.



\begin{lemma}
	\label{lem:pc}
	\begin{equation}
	PC(s_i) \leq PC(s_1 \vee s_2),\ where\ i = 1, 2
	\end{equation}	
    \begin{equation}
	PC(s_i) \geq PC(s_1 \wedge s_2),\ where\ i = 1, 2
	\end{equation}	
\end{lemma}

\begin{proof}	
	Given $i = 1$, for any true positive record pair $t \in B_{s_1}$, we have $t \in B_{s_1} \cup B_{s_2} =  B_{(s_1 \vee s_2)}$. This is to say, the number of true positives generated by scheme either $s_1$ or $s_2$ cannot be larger than that generated by scheme $s_1 \vee s_2$ and cannot be smaller than that generated by scheme $s_1 \wedge s_2$, i.e. $|tp(B_{s_1})| \leq |tp(B_{s_1 \vee s_2})|$ and $|tp(B_{s_1})| \geq |tp(B_{s_1 \wedge s_2})|$. Since we know the sum of true positives and false negatives is constant, which refers to the total number of matched pairs in the dataset, based on the definition of PC, we prove the lemmas.    

\end{proof}

The lemma states that, if a scheme generates blocks which discard more matches than expected, a disjunction form of extension can be applied in order to increase matches contained in blocks. On the contrary, if a conjunction form of extension is applied, the completeness of the matches will be reduced.

\subsection{Grid-based Naive Learning}
\label{sec_gn}
\begin{figure}
	\begin{center}
		\includegraphics[height = 0.23 \textheight]{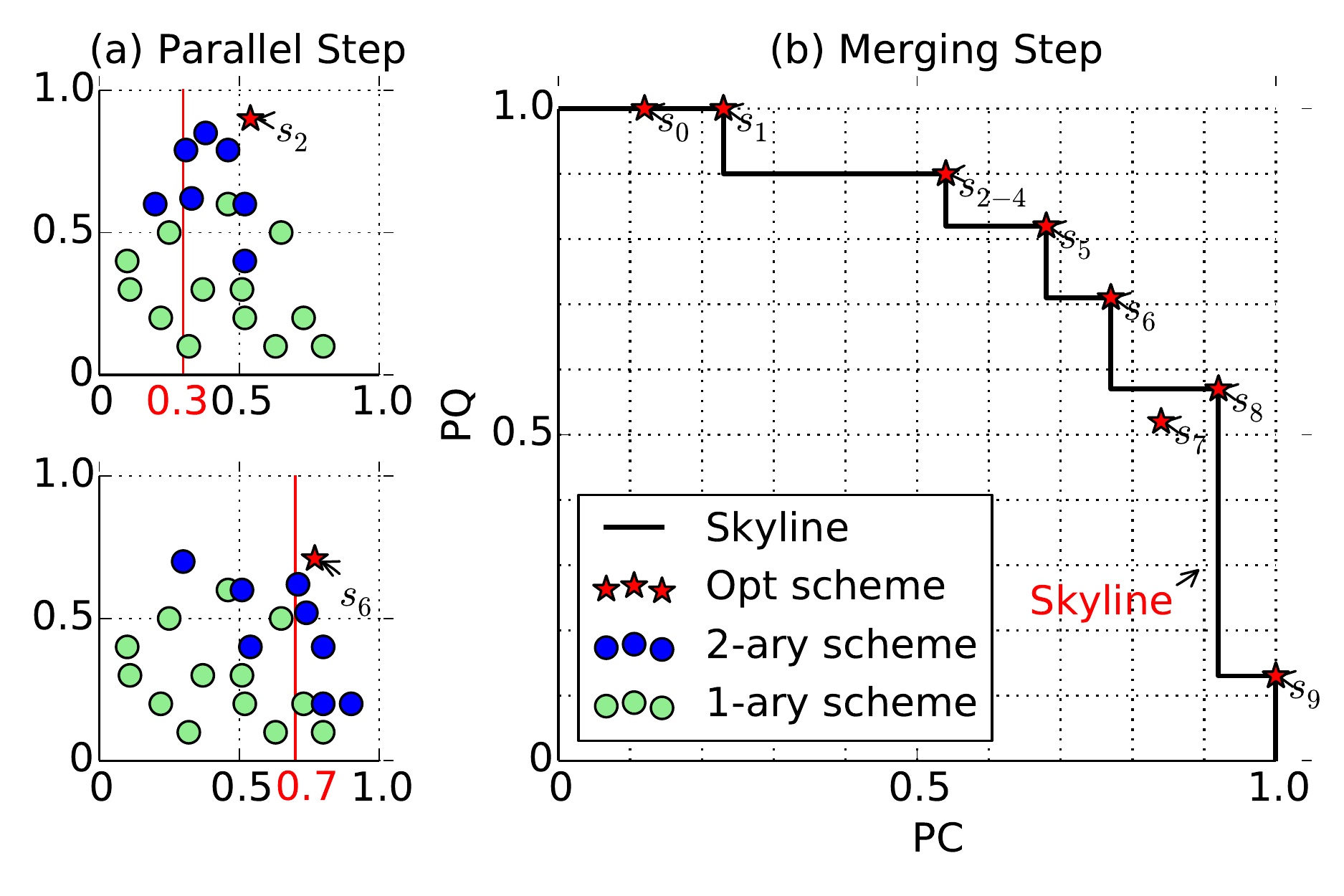}
	\end{center}
\vspace{-5mm}
	\caption{An illustration of the grid-based naive learning algorithm, where optimal blocking schemes are learned in parallel as shown in (a), and the scheme skyline is depicted in (b), where $s_7$ is dominated by the skyline, and $s_{2}-s_{4}$ are the same scheme shown as $s_{2-4}$.}
	\label{fig:naivegrid}	
\end{figure}

A naive way to learn skyline blocking schemes is to learn the optimal blocking schemes under different thresholds in one or more dimensions. Then, based on these optimal schemes, a scheme skyline can be learned.

In the following, we first present an \emph{Active Scheme Learning (ASL)} algorithm which can actively learn the optimal blocking scheme under a given PC threshold and a specified label budget. In this algorithm, the concept of \emph{optimal blocking scheme} is defined as follows.

\begin{definition}
\label{def:opt}
	(Optimal blocking scheme) Given a human oracle $\zeta$, and a PC threshold $\epsilon \in [0, 1]$, the \emph{optimal blocking scheme} is a blocking scheme $s_{\epsilon, \zeta}$ w.r.t. the following objective function, through selecting a training set $T$:
	\label{prob_1}
	\begin{equation}
	\label{eq:opt}
	\begin{aligned}
	& \underset{}{\text{\textbf{maximize}}}
	& & PQ  \\
	& \text{\textbf{subject to}}		
	& & PC \geq \epsilon,\ \text{and }\ |T| \leq budget(\zeta) \\		
	&&& 	
	\end{aligned}
	\end{equation}
    Additionally, a \emph{locally optimal blocking scheme} over $S$ is the optimal scheme among a set $S$ of schemes which satisfies the above criteria.
\end{definition}

At its core, the ASL algorithm adopts two active learning strategies: \emph{active sampling} aims to tackle the class imbalance problem by solving the balanced sampling problem and thus to reduce the label cost, \emph{active branching} is a scheme extension strategy that is targeted to reduce the searching space and thus further reduce the label cost.

\begin{algorithm}
	\KwIn{Dataset: $D$ \\
		\hspace{9mm}	PC threshold $\epsilon \in (0, 1]$\\
		\hspace{9mm}	Human oracle $\zeta$\\
		\hspace{9mm}	Set of blocking predicates $P$\\
		\hspace{9mm}	Sample size $k$\\
	}
	\KwOut{
		A blocking scheme $s$
	}
	$S = S_{prev} = P$, $n = 0$, $T = \emptyset$,	$X = \emptyset$ \\
	$X = X \cup \textsc{Random\_sample}(D)$ \\
	\While{$n < budget (\zeta)$}{		
		\For
        {each $s_i \in S$}{	
			\If{$\gamma(s_i, X) \leq 0$}{
				$X = X \cup \textsc{Similar\_sample}(D, s_i, k)$   \\
			}
			\Else{	
				$X = X \cup \textsc{Dissimilar\_sample}(D, s_i, k)$ \\
			}	
			$n = |X|$ 
		}
		$T = T \cup \{(x_i, \zeta(x_i)) | x_i \in X\}$	\\
		$s = \textsc{Find\_optimal\_scheme} (S, T, \epsilon)$; $S_{prev} = S$ \\
		\If{$\textsc{Found}(s)$}{ 			
			$S = \{s \wedge s_i| s_i \in S_{prev} \}$
		}
		\Else{
			$s = \textsc{Find\_approximate\_scheme} (S, T, \epsilon)$	\\ 
			$S = \{s \vee s_i| s_i \in S_{prev} \}$	
		}
	}	
	\textbf{Return} $s$
	\caption{Active Scheme Learning (ASL)}
	\label{Algo:A1}
	
\end{algorithm}

A high-level description for the algorithm is presented in Algorithm~\ref{Algo:A1}. Let $S$ be a set of blocking schemes, where each blocking scheme $s_i \in S$ is a blocking predicate at the beginning. The budget usage is initially set to zero, i.e. $n = 0$.
A set of feature vectors is selected from the dataset as seed samples (lines 1 and 2). Then, the algorithm iterates until the number of samples in the training set reaches the budget limit (line 3). At the beginning of each iteration, the active sampling strategy is applied to generate a training set (lines 4 to 10). For each blocking scheme $s_i \in S$, the samples are selected in two steps: (1) firstly, the balance rate of this blocking scheme $s_i$ is calculated (lines 5 and 7), (2) secondly, a feature vector to reduce this balance rate is selected from the dataset (lines 6 and 8). Then the samples are labeled by the human oracle and stored in the training set $T$. The usage of label budget is increased, accordingly (lines 9 and 10).

A locally optimal blocking scheme $s$ is searched among a set of blocking schemes $S$ over the training set, according to a specified PC threshold $\epsilon$ (line 11). Then the active branching strategy is applied to determine whether a conjunction or disjunction form of this optimal blocking scheme and other blocking schemes will be used based on Lemma~\ref{lem:pc} (lines 12-16). If the locally optimal blocking scheme is found, new blocking schemes are generated by extending $s$ to a conjunction with each of the blocking schemes in $S_{prev}$, so that the PQ of new blocking schemes may be further increased (lines 12 and 13). Otherwise a blocking scheme with the maximum PC value is selected and new schemes are generated using disjunctions, so that the PC of new schemes can be further increased (lines 14 to 16).

Now, we propose the \emph{Naive Learning} algorithm for scheme skylines, which is based on the idea of using ASL for learning optimal blocking schemes under different PC thresholds. The naive learning algorithm is described in Algorithm~\ref{Algo:GN}. The input a user need to specify is the total label budget, the maximal ary of schemes in scheme skyline, and the size $\Delta$ of threshold interval, i.e. the difference of two consecutive thresholds. For example, if $\Delta$ is set to be 0.1, there will be in total 10 thresholds used, i.e. 0.1, 0.2, ..., 1.0, to learn at most 10 optimal blocking schemes. It is possible that the same blocking scheme is learned under two different thresholds. As shown in Algorithm~\ref{Algo:GN}, this approach consists of two steps. In the first step, called \emph{parallel step}, the threshold interval $\Delta$ is defined (line 1), and the optimal blocking schemes are learned in parallel under each threshold using the ASL algorithm (lines 2-5). Here, in our algorithm, sample size $k$ is uniformly decided in terms of the total label budget $budget(\zeta)$, the threshold interval $\Delta$ and the number of predicates $|P|$. In the second step, called \emph{merging step}, all optimal blocking schemes are merged and their domination is checked, which generates the scheme skyline (lines 6-7). Fig.~\ref{fig:naivegrid}(a) illustrates that two optimal blocking schemes are learned in parallel, where their thresholds are set to be 0.3 and 0.7, respectively. In Fig.~\ref{fig:naivegrid}(b), all the optimal schemes that are learned under the thresholds 0.1, 0.2, \dots, 1.0 are merged together and the scheme skyline is generated.

\begin{algorithm}
	\KwIn{Dataset: $D$ \\
		\hspace{9mm}	Human oracle $\zeta$\\
		\hspace{9mm}	Set of blocking predicates $P$\\
		\hspace{9mm}	Size of threshold interval $\Delta$\\
        \hspace{9mm}	Sample size $k$\\
	}
	\KwOut{
		Scheme Skyline $S^*$
	}
	$\epsilon = \Delta$, $S^* = \emptyset$, $S = \emptyset$\\
    \While{$\epsilon \leq 1$}{	
      $s = \textsc{ASL} (D, \epsilon, \zeta, P, k)$\\
      $S = S \cup \{s\}$\\
      $\epsilon += \Delta$\\
    }   
    $S^* = \textsc{Skyline\_schemes} (S)$\\
	
	\textbf{Return} $S^*$
	\caption{Grid-based Naive Learning (Naive-Sky)}
	\label{Algo:GN}
	
\end{algorithm}

\subsection{Grid-based Active Learning}
\label{sec_ga}


There are some limitations in the grid-based naive learning approach. For example, how to choose an appropriate size of threshold interval? If $\Delta$ is set too high,  blocking scheme points in the skyline will be sparse. On the contrary, if $\Delta$ is set too low, users may have dense blocking scheme points in the skyline, but a large parallel iteration step will be involved in learning optimal blocking schemes under different thresholds, and some of them are often redundant. This leads to unnecessary computational costs.  
\begin{figure} 
	\begin{center}
		\includegraphics[height = 0.11 \textheight]{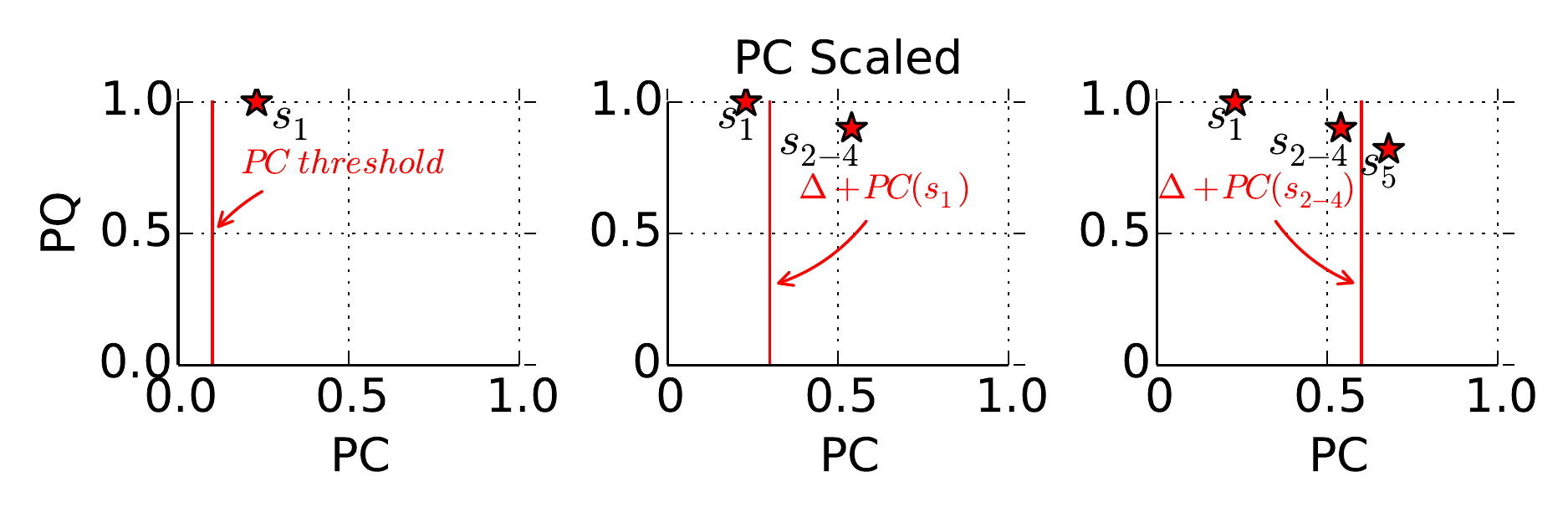}
	\end{center}
\vspace{-5mm}
	\caption{Grid-based active learning algorithm which can actively choose PC threshold. We use the same example as presented in Fig.~\ref{fig:naivegrid}.}
	\label{fig:ga}	
\end{figure}

During our study with the grid-based naive learning algorithm, we have noticed that the PC threshold specified by a user may not be (or even far from) the actual PC value of the optimal blocking scheme learned by the grid-based naive learning algorithm. For example, when the threshold is set to 0.3, the actual PC value of the learned optimal blocking scheme can be 0.53. Moreover, when the threshold is set to 0.4, e.g. $\Delta = 0.1$, the learned optimal blocking scheme still remains unchanged. To leverage this observation, we propose another algorithm for learning scheme skylines, called \emph{Grid-based Active Learning}, to efficiently learn the scheme skyline even for a small $\Delta$ based on the following lemma. The process of this algorithm is shown in Fig.~\ref{fig:ga}.

\begin{lemma}\label{lem:cr}
  
  Assume $s_{\epsilon}$ and $s_{\epsilon'}$ are two optimal schemes learned by ASL under the PC thresholds $\epsilon$, $\epsilon'$ and the same label budget. The following property holds:
  \begin{equation}
    \forall \epsilon' \in [\epsilon, PC(s_{\epsilon})], \ s_{\epsilon'} = s_{\epsilon}.
  \end{equation}
	
  
\end{lemma}

\begin{proof}
	Given a threshold $\epsilon$, a set of blocking schemes $S$ ($\forall s \in S$, $PC(s) > \epsilon$), we have $s_\epsilon \in S$. Given any threshold $\epsilon'$, s.t. $\epsilon < \epsilon' \leq PC(s_\epsilon)$, we have a set of blocking schemes $S'$ ($\forall s' \in S'$, $PC(s') > \epsilon'$ ). We can also have $s_\epsilon \in S'$ and $s_{\epsilon'} \in S'$. Based on Definition~\ref{def:opt}, the optimal blocking scheme is unique. we have $s_\epsilon = s_{\epsilon'}$.
\end{proof}

The grid-based learning algorithm is presented in Algorithm~\ref{Algo:GA}. By Lemma~\ref{lem:cr}, the algorithm is designed to actively select the next PC threshold based on the PC value of current optimal blocking scheme (line 5).
 
\begin{algorithm}
	\KwIn{Dataset: $D$ \\
		\hspace{9mm}	Human oracle $\zeta$\\
		\hspace{9mm}	Set of blocking predicates $P$\\
		\hspace{9mm}	Size of threshold interval $\Delta$\\
        \hspace{9mm}	Sample size $k$\\
	}
	\KwOut{
		Skyline blocking schemes $S^*$
	}
	$\epsilon = \Delta$, $S^* = \emptyset$, $S = \emptyset$ \\
    \While{$\epsilon \leq 1$}{	      
      $s = \textsc{ASL} (D, \epsilon, \zeta, P, k)$\\
      $S = S \cup \{s\}$\\
      $\epsilon = \textsc{PC}(s) + \Delta$\\
    }    
    $S^* = \textsc{Skyline\_schemes} (S)$\\
	\textbf{Return} $S^*$
	\caption{Grid-based Active Learning (Active-Sky)}
	\label{Algo:GA}
	
\end{algorithm}



\subsection{Progressive Skyline Learning}


\begin{figure*} 
	\begin{center}
		\includegraphics[height = 0.19 \textheight]{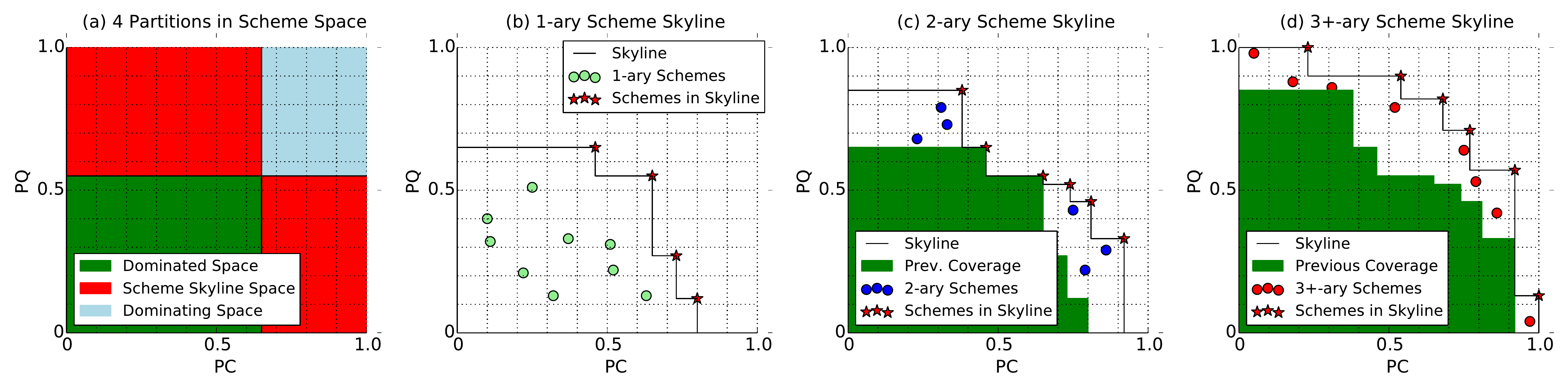}
	\end{center}
\vspace{-5mm}
	\caption{Progressive learning algorithm in progress: (a) shows the space separated by PC and PQ that we may find scheme skyline (blue and red area) or dominated by current skyline point (green area); (b) - (d) present the progress of our algorithm. }
	\label{fig:progressive}	
\end{figure*}

The grid-based active learning approach has solved the problem that the lower the $\Delta$ is set, the more iterations (each iteration learns one optimal blocking scheme) the ASL algorithm has to take. However, it still has to construct the training set for each iteration, and distribute the label budget into different iterations for sampling. This gives rise to a new question: can we reduce the iteration times to further reduce the label cost?

Here we propose a \emph{Progressive Skyline Learning} approach. This approach can learn a scheme skyline progressively, i.e., it starts by learning the 1-ary scheme skyline and ends by learning the n-ary scheme skyline, as illustrated in Fig~\ref{fig:progressive} In this approach, a user does not need to set the size of threshold interval $\Delta$. 

Given a blocking scheme, we can partition the scheme space into four parts, as shown in Fig.~\ref{fig:progressive}(a). Blocking schemes in the green area, called \emph{Dominated Space}, are dominated by the given blocking scheme. Blocking schemes in the skyline can be found in the red area, called \emph{Scheme Skyline Space}. Blocking schemes in the blue area, called \emph{Dominating Space}, dominate the given blocking scheme. Clearly, if a blocking scheme is in the skyline, its dominating space should have no any other schemes. 
Therefore, our objective for progressive skyline learning is: given a scheme skyline $S^*$, we try to extend each $s \in S^*$ by discarding schemes in its dominated space, verifying schemes in its scheme skyline space, and finding schemes in its dominating space.

The algorithm is described in Algorithm~\ref{Algo:PL}, which does not change much compared with the previous algorithms in the sampling steps (lines 1-11). However, in the scheme extension step, we take the property illustrated in Fig.~\ref{fig:progressive}(a) into consideration. For each scheme in the skyline, we extend it with all the other predicates that are not included in this scheme (lines 12-13). We extend the existing schemes in $S^*$ by conducting both conjunctions and disjunctions with all predicates, obtain both PC and PQ of the new schemes in terms of the training set and select ones with incremental PC or PQ (line 14). That is to say, if the new scheme fails in the red area as shown in Fig.~\ref{fig:progressive}(a), we add this scheme into the candidate set. If it appears in the blue area, we replace the previous one by this new scheme. If it appears in the green area, we discard it.
\begin{algorithm}[h]
	\KwIn{Dataset: $D$ \\
		\hspace{9mm}	Human oracle $\zeta$\\
		\hspace{9mm}	Set of blocking predicates $P$\\
		\hspace{9mm}	Ary of Schemes $l$\\
	}
	\KwOut{
		scheme skyline $S^*$
	}
	$S = P$, $n = 0$, $T = \emptyset$,	$X = \emptyset$, $k = \frac{budget(\zeta)}{2l \times |P|} $\\
	$X = X \cup \textsc{Random\_sample}(D)$ \\
	\While{$n < budget (\zeta)$}{		
		\For(\tcp*[f]{Begin sampling}){each $s_i \in S$}{	
			\If{$\gamma(s_i, X) \leq 0$}{
				$X = X \cup \textsc{Similar\_sample}(D, s_i, k)$   \\
			}
			\Else{	
				$X = X \cup \textsc{Dissimilar\_sample}(D, s_i, k)$ \\
			}	
			$n = |X|$ \tcp*[f]{End sampling}
		}
		$T = T \cup \{(x_i, \zeta(x_i)) | x_i \in X\}$	\tcp*[f]{Add samples}\\
		$S^* = \textsc{Scheme\_skyline} (S)$; $S = \emptyset$ \\
        \For{$s_j \in S^*$}{
          \For{$p_i \in P$}{
          $\textsc{Select\_Schemes}(p_i, s_j, S)$  
          }
        }
	}	
	\textbf{Return} $S^*$
	\caption{Progressive Learning (Pro-Sky)}
	\label{Algo:PL}
	
\end{algorithm}




\begin{table*}
\caption{Characteristics of datasets}
	\centering
	\label{tab:dataset}
	\begin{tabular}{cccccc}
		\toprule
		Datasets  &$\#$ Attributes&$\#$ Records& $\#$ True Matches &Class Imbalance Ratio & \# Blocking Predicates\\
		\midrule
		Cora & 4 & 1,295 & 17,184 &1 : 49 & 16\\
		DBLP-Scholar & 4/4 &  2,616 / 64,263 & 2360 & 1 : 71,233 & 16/16\\
		DBLP-ACM & 4/4 &  2,616 / 2,294 & 2,224 & 1 : 2,698 & 16/16\\		
		NCVR & 18/18 &  6,233,785 / 6,981,877 & 6,122,579 & 1 : 6.6 $\times$ $10^6$ & 72/72\\
		\bottomrule
	\end{tabular}
\end{table*}

\subsection{Complexity Analysis}\label{subsec-complexity}

In this section, we discuss the search complexity and the time complexity for learning the scheme skyline.

\subsubsection{Search complexity}
Given $n$ blocking predicates, we have $2^n$ blocking schemes composed of blocking predicates in conjunctions. Furthermore, if a blocking scheme is composed of at most $n$ different blocking predicates (i.e. n-ary blocking scheme) in disjunction of conjunctions, we can regard them as the \emph{monotonic boolean functions}, which are defined to be the expression combining the inputs (which may appear more than once) using only the operators conjunction and disjunction (in particular "not" is forbidden) \cite{kisielewicz1988solution}.
Hence the searching complexity for all possible blocking schemes can be $O (2^{n \choose [n/2]} )$ asymptotically, which is also known as the \emph{Dedekind Number}. Learning a scheme skyline in this way will be no doubt to be accurate, because all blocking schemes will be considered. However, this is space consuming and label wasting when the number of blocking predicates is large, but the number of blocking schemes in a scheme skyline is small.

To analyze the search complexity of schemes of the algorithms \emph{Grid-based Naive Learning} and \emph{Grid-based Active Learning}, we first analyze the complexity of ASL. Given $n$ blocking predicates, with sufficient label budget, in the worst case, we can learn a n-ary blocking scheme as output. During this process, we first need to search for all $n$ blocking schemes. Then based on the locally optimal one, we need to search for $n-1$ blocking schemes of 2-ary, and then, $n-2$ blocking schemes of 3-ary. Accordingly, the searching complexity of this process is $O(n^2)$. Given $\Delta$, the total complexity of \emph{Grid-based Naive Learning} is $O(\frac{n^2}{\Delta})$, and in the worst case, the the total complexity of \emph{Grid-based Naive Learning} is the same as $O(\frac{n^2}{\Delta})$. The complexity of algorithm \emph{Progressive Learning}, if we use \emph{$|Opt_i|$} to describe the number of schemes selected in the i-ary, the total complexity will be $\sum_{i = 1}^{n} |Opt_i| \times 2|P|$, where 2 indicates both conjunction and disjunction of two schemes.In the worst case that all schemes are in the skyline, is the same as \emph{Dedekind Number}. However, if $\Delta$ is introduced to this algorithm, i.e., there should be at least a distance of $\Delta$ between two schemes in skyline, the total complexity will be $\frac{n^2}{\Delta}$ ($\frac{1}{\Delta} \leq n$), in the worst case.

\subsubsection{Time complexity}

We further discuss the time complexity of sampling. To tackle the class imbalance problem, we select both similar and dissimilar samples w.r.t. a given blocking scheme $s$. However, as explained in Section~\ref{sec_as}, it may be a high imbalance ratio of similar and dissimilar samples or even impossible to select one w.r.t. a blocking predicate. In the worst case, we have to traverse the whole dataset to obtain one sample. Hence the time complexity to generate $k$ samples will be $O(|D| \times k)$, where $D$ is a dataset and $k$ is the sample size for one predicate.

To make the algorithms efficient under a large sample budget, we have adopted index tables for each candidate scheme by obtaining their values, hence the algorithm can choose either similar or dissimilar samples it needs in linear time. In this way, the time complexity will be $O(|D| + k)$ for any candidate scheme.

%% file: Evaluation.tex
\section{Evaluation}\label{sec_eva}


We have evaluated our algorithms to experimentally verify their performance. All our experiments have been run on a server with 6-core 64-bit Intel Xeon 2.4 GHz CPUs, 128GBytes of memory.

\subsection{Experimental Setup}

We present the datasets, blocking predicates, baseline approaches and measures used in our experiments.


\subsubsection{Datasets} We have used four datasets in our experiments: (1) \emph{Cora}\footnote[1]{Available from: \emph{http://secondstring.sourceforge.net}} dataset contains bibliographic records of machine learning publications. 
(2) \emph{DBLP-Scholar}\footnotemark[1] dataset contains bibliographic records from the DBLP and Google Scholar websites.
(3) \emph{DBLP-ACM}\cite{kopcke2010evaluation} dataset contains bibliographic records from the DBLP and ACM websites.
(4) \emph{North Carolina Voter Registration (NCVR)}\footnote[2]{Available from: \emph{http://alt.ncsbe.gov/data/}} dataset contains real-world voter registration information of people from North Carolina in the USA. Two sets of records collected in October 2011 and December 2011 respectively are used in our experiments.
We summarize the characteristics of these data sets in Table~\ref{tab:dataset}. A complete list of attributes in these data sets are presented in Table \ref{tab:dataset-attributes}. 
\begin{table}

	\caption{Attributes of datasets}
	\centering
	\label{tab:dataset-attributes}
	\begin{tabular}{|c|c|}
		\toprule
		Datasets  & Attributes\\
		\midrule
		\multirow{3}{*}{Cora} & authors, pub\_details, title, affiliation,\\ &conf\_journal, location, publisher, year,\\ &pages, editors, appear, month  \\\hline
		DBLP-Scholar & \multirow{2}{*}{title, authors, venue and year} \\\cline{1-1}
		DBLP-ACM & \\\hline		
		\multirow{6}{*}{NCVR} & county\_id, county\_desc, voter\_reg\_num,  \\
        & voter\_status\_desc,  voter\_status\_reason\_desc,\\ &  absent\_ind, last\_name, first\_name, midl\_name, \\&full\_name\_rep, full\_name\_mail, reason\_cd,\\ & status\_cd, house\_num, street\_name, \\ &street\_type\_cd, res\_city\_desc and state\_cd\\
		\bottomrule
	\end{tabular}
\end{table}

\subsubsection{Blocking predicates}
We have used the following blocking functions in our experiments:

\begin{itemize}
\item \emph{Exact-Match:} This blocking function takes two strings as input and compares the string values. The function returns true if the strings are exactly the same.

\item \emph{Soundex:} This blocking function takes two strings as input and transfers them into soundex codes based on a reference table \cite{holmes2002improving}. Then it returns true if two codes are the same.

\item \emph{Double-Metaphone:} This blocking function also takes two strings as input and transfers each string into two codes from two reference tables \cite{philips2000double}. It returns true if either code is the same. Comparing with \emph{Soundex}, it has a better performance when it is applied to Asian names. 

\item \emph{Get-substring:} This blocking function takes only a segment (i.e. first four letters) of the whole string for comparison, and returns true if the segments from two strings are the same.

\end{itemize}
The above blocking functions have been applied to all attributes in the datasets depicted in Table \ref{tab:dataset-attributes}, which accordingly leads to 16 or 72  blocking predicates in each dataset as shown in Table \ref{tab:dataset}.

\subsubsection{Baseline approaches} 
We have used the following approaches as the baselines: (1) \emph{Fisher} \cite{kejriwal2013unsupervised}, which is the state-of-the-art unsupervised scheme learning approach proposed by Kejriwal and Miranker. Details of this approach will be outlined in Section~\ref{sec_rw}.  (2) \emph{TBlo} \cite{fellegi1969theory}, which is a traditional blocking approach based on expert-selected attributes. In the survey \cite{christen2012survey}, this approach has a better performance than the other approaches in terms of the F-measure results. (3) \emph{RSL (Random Scheme Learning)}, which is an algorithm that is similar to the structure of the ASL algorithm, but uses random sampling, instead of active sampling, to build a training set and learn blocking schemes. In each experiment, we have run the RSL ten times. We present the average results of the blocking schemes it has learned. 

\subsubsection{Measures} We use the following measures \cite{christen2012survey} to evaluate the blocking quality of our approach. \emph{Reduction Ratio (RR)} is one minus the total number of record pairs in blocks divided by the total number of record pairs without blocks, i.e., RR measures the reduction of the number of compared pairs with and without blocks. \emph{Pairs Completeness (PC)} and \emph{Pairs Quality (PQ)} have been defined in Section~\ref{sec_pd}. \emph{F-measure} $FM = \frac{2*PC*PQ}{PC + PQ}$ is the harmonic mean of PC and PQ. 

The size of a training set directly affects the scheme skyline learning results. If the training set is small, an algorithm may learn different skylines in different runs. This is because the undersampling can lead to biased samples which do not represent the characteristics of the whole sampling space. Hence, we define the notion of \emph{constraint satisfaction} as $CS = \frac{N_s}{N}$ to describe the learning stability of an algorithm.  We use $N_s$ to denote the number of times that an algorithm can learn a blocking scheme, and $N$ is the total number of times the algorithm runs. For example, if an algorithm runs ten times and learns three different scheme skylines with 2, 3, and 5 times, respectively, then the CS value for the third blocking scheme is $\frac{5}{10} = 0.5$ in this case.

\medskip
In the rest of this section, we will present the experimental results of our skyline algorithms, and compare the performance of our approach with the baseline approaches.

\subsection{Label Efficiency}

\begin{figure*} [ht]
	\begin{center}
		\includegraphics[width = \textwidth]{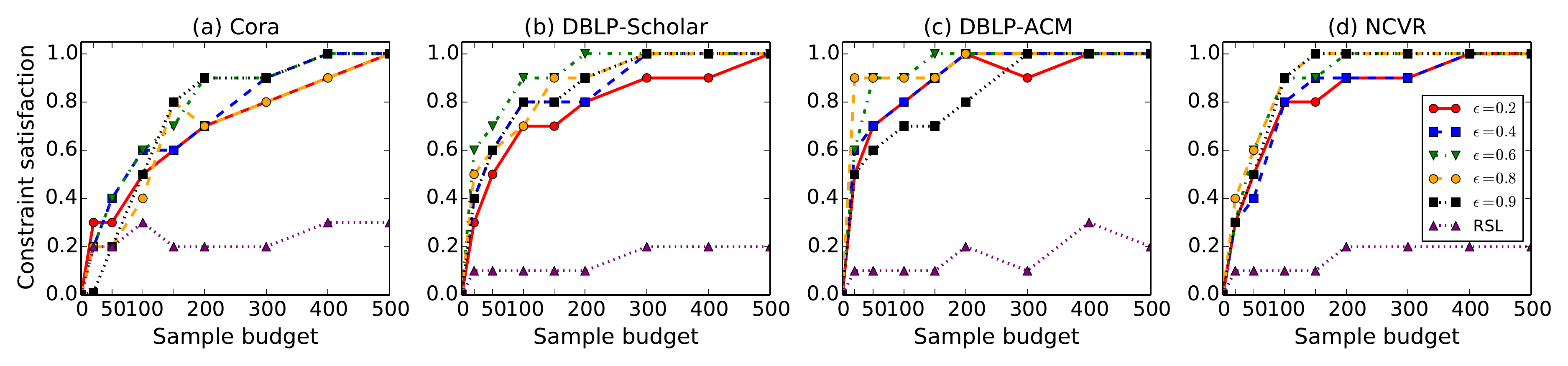}
	\end{center}
    \vspace{-5mm}
	\caption{Comparison on constraint satisfaction by ASL and RSL under different label budgets over four datasets} 
	\label{fig:budget}
\end{figure*}

We have evaluated the label efficiency of our algorithms.

\subsubsection{Label cost}
The label cost of our proposed algorithms, namely \emph{Naive-Sky} for algorithm~\ref{Algo:GN}, \emph{Active-Sky} for algorithm~\ref{Algo:GA} and \emph{Pro-Sky} for algorithm~\ref{Algo:PL}, are shown in Table~\ref{tab:skybudget}, where the label costs of each algorithm for learning the scheme skylines of four datasets with $CS = 90\%$ are recorded. We have defined $\Delta = 0.1$. For \emph{Naive-Sky}, the label budget begins with 50, and increases by 50 for each run of ASL. For \emph{Active-Sky}, the budget is the same. The total label cost is accumulated when the PC threshold increases. For \emph{Pro-Sky}, the label budget begins with 500, and increases by 500. The label cost is recorded when CS reaches 90\% in ten runs. Table~\ref{tab:skybudget} shows that  \emph{Active-Sky} saves more labels than \emph{Naive-Sky}; nonetheless, \emph{Pro-Sky} can reduce the label cost from one third to a half of the label usage of \emph{Active-Sky}.

\begin{table} [h]
	\caption{Comparison on label cost of scheme skyline algorithms \emph{Naive-Sky}, \emph{Active-Sky} and \emph{Pro-Sky} with CS = 90\%}	
	\centering
	\label{tab:skybudget}
	\begin{tabular}{|c|c|c|c|c|}
		\toprule
		Algorithm&Cora&DBLP-Sc.&DBLP-ACM&NCVR\\
		\hline
		Naive-Sky & 6000 & 5000 & 3000 & 3500\\
		\hline
		Active-Sky & 4200 & 3500 & 1250 & 1400\\
		\hline
		Pro-Sky & 2500 & 2000 & 1000 & 1000\\	
		\hline			
	\end{tabular}\\
	
\end{table}

\begin{table} [h]
	\caption{Comparison on label cost of ASL and RSL with CS = 90\%}	
	\centering
	\label{tab:budgetsize}
	\begin{tabular}{|c|c|c|c|c|}
		\toprule
		$\epsilon$ &Cora&DBLP-Scholar&DBLP-ACM&NCVR\\
		\hline
		0.2 & 600 & 500 & 300 & 300\\
		\hline
		0.4 & \textbf{400} & 350 & 200 & 350\\
		\hline
		0.6 & 450 & \textbf{250} & \textbf{150} & 250\\	
		\hline	
		0.8 & 550 & 300 & 200 & \textbf{200}\\
		\hline
		0.9 & 500 & \textbf{250} & 300 & 250\\
		\hline
		RSL & 7,900 & 10,000+ & 2,200 & 10,000+\\
		\hline
	\end{tabular}\\
	
\end{table}

\begin{table*}
	\vspace{-5mm}
	\caption{Comparison on the running time (in seconds) of different algorithms over four datasets}
	\centering
	\label{tab:runtime}
    
    \begin{tabular}{|c|ccc|ccc|ccc|ccc|}
		\toprule
		&\multicolumn{3}{|c|}{Cora}&\multicolumn{3}{|c|}{DBLP-Scholar}&\multicolumn{3}{|c|}{DBLP-ACM} & \multicolumn{3}{|c|}{NCVR} \\
        \cline{2-13}
        &Budget & $\Delta$& RT &Budget & $\Delta$& RT&Budget & $\Delta$& Run Time&Budget & $\Delta$& RT\\
		\hline
		\multirow{2}{*}{Naive-Sky} &\multirow{2}{*}{6000}&0.1& 18.74 &\multirow{2}{*}{5000}&0.1& 78.18 &\multirow{2}{*}{3000}&0.1& 76.95 &\multirow{2}{*}{3500}&0.1& 225.08 \\
        &&0.05&21.55&&0.05&89.9&&0.05&90.8&&0.05&267.85\\
		\hline
        \multirow{2}{*}{Active-Sky} &\multirow{2}{*}{4200}&0.1& 12.48 &\multirow{2}{*}{3500}&0.1& 56.28 &\multirow{2}{*}{1250}&0.1& 33.23 &\multirow{2}{*}{1400}&0.1& 118.4 \\
        &&0.05&14.68&&0.05&60.73&&0.05&34.55&&0.05&122.16\\
		\hline
        \multirow{2}{*}{Pro-Sky} &\multirow{2}{*}{2200}&0.1& 5.6 &\multirow{2}{*}{1600}&0.1& 32.56 &\multirow{2}{*}{800}&0.1& 11.38 &\multirow{2}{*}{800}&0.1& 63.09 \\
        &&0.05&6.27&&0.05&38.09&&0.05&13.28&&0.05&70.56\\
		\hline
	\end{tabular}    
\end{table*}
Both the \emph{Naive-Sky} and \emph{Active-Sky} algorithms use ASL to build scheme skylines. In order to analysis the factors that affect the label cost and to show that our active sampling strategy has better performance than random sampling, we have evaluated the label efficiency of ASL algorithm. Here, We present the label costs for learning a stable blocking scheme with CS = 90\% under different PC thresholds, and compare them with the number of labels required by RSL (random sampling algorithm) in Table~\ref{tab:budgetsize}. The minimum label cost for each dataset is marked in black. In our experiments, the label budget for ASL under a given PC threshold is tested with 50 at the beginning, and then increased by 50. The label budget for RSL is tested with 50 at the beginning, and increased by 50 each time. The experiments for both ASL and RSL terminate when the CS values of the learned blocking schemes reach 90\% in ten consecutive runs.

\subsubsection{Constraint satisfaction evaluation}
\label{sec_cs}
To further analyze the label efficiency of our algorithms, especially the naive-sky and the active-sky algorithms, we have monitored the constraint satisfaction under different label budgets (ranging from 20 to 500) in terms of various PC thresholds (e.g. $\epsilon \in \{0.1, 0.2, 0.4, 0.6, 0.8\}$) over four datasets. The results are presented in Fig.~\ref{fig:budget}.

We use the total label budget as the training label size for RSL to make a fair comparison on active sampling and random sampling. Our experimental results show that random sampling with a limited number of labels fails to identify an optimal blocking scheme. Additionally, both PC threshold and label budget can affect the constraint satisfaction. In general, when the label budget increases or the PC threshold $\epsilon$ decreases, the CS value grows. It is worthy to note that an extremely high PC threshold is usually harder to achieve, and thus sometimes no scheme that satisfies the threshold can be learned due to the limitation of samples (e.g. the black line with square marks under 100 sample budget). However, on the other hand, if the PC threshold is set extremely low under a low label budget (e.g. the red line with $\epsilon = 0.2$), it could generate quite a number of blocking schemes satisfying the threshold with different PQ , which also leads to a low CS value. This is happened because that in our approach, only the optimal scheme is selected as candidate. With a limited number of samples, a sub-optimal blocking scheme may be selected.

\subsubsection{Label efficiency analysis}

Here we analyze the factors that can affect the label cost. First, as explained before, both extremely high and low PC thresholds may cost more labels than other cases. Second, datasets of a smaller size (i.e., a smaller number of record pairs) often need less labels such as DBLP-ACM. Furthermore, datasets with more attributes (e.g. NCVR) have a larger number of blocking predicates and thus need a higher label budget for sampling all its blocking predicates. The quality of a dataset may also affect the label cost. Generally, the cleaner a dataset is, the less labels it costs. For example, even though Cora has the smallest size of attributes and records, it still needs the largest number of labels because it contains fuzzy values (e.g. mis-spelling names, exchanged first name and last name). On the contrary, NCVR needs a lower label budget even it has more attributes and records. The cleanness of a dataset also affects the distribution of label costs under different thresholds. For example, with a PC threshold $\epsilon = 0.2$, we need only 200 labels for NCVR but 550 labels for Cora to generate blocking schemes with CS = 90\%.

\subsection{Time Efficiency} 
We show the Run Time (RT) of our three algorithms over four datasets in Table~\ref{tab:runtime}. The label budget we use is decided by Table~\ref{tab:skybudget}, where the algorithms can learn consistent results as output. Two threshold intervals are used, i.e. $\Delta = 0.05$ and $\Delta = 0.1$. Generally speaking, the run time depends on two factors: (1) the size of the dataset and (2) the number of samples we need. With the incremental of the label budget and the number of records in one dataset, the run time grows. For the algorithms, the threshold interval $\Delta$ may affect the run time, but not significantly.

\begin{figure*}[h]	
	\centering
		\includegraphics[width = \textwidth]{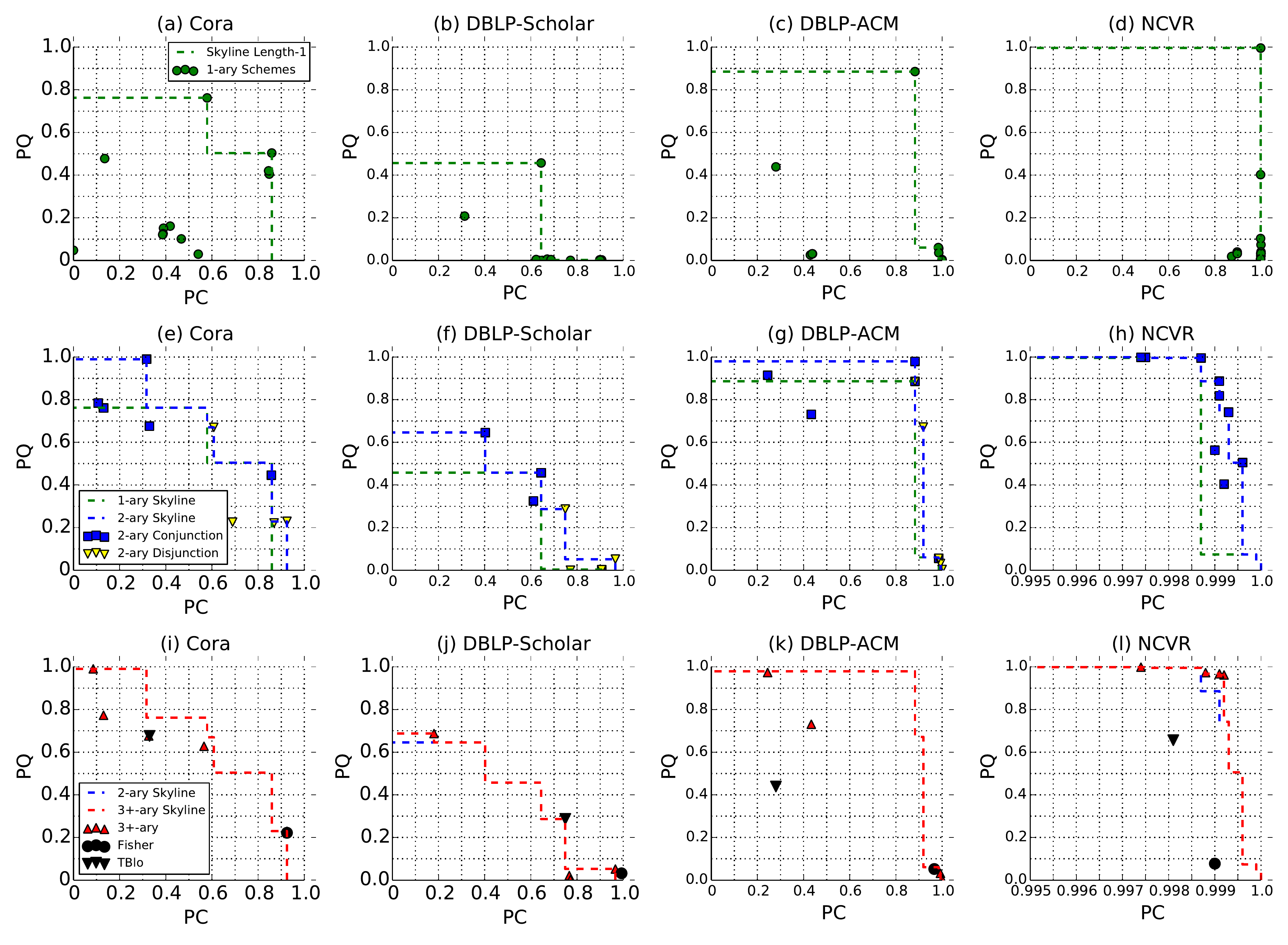}
	\caption{The progressive process for learning scheme skylines by \emph{Pro-Sky} over four datasets}
	\label{fig:Skyline_ex}
\end{figure*}
\subsection{Blocking Quality} 
Now we discuss the blocking quality of our algorithms.
\begin{figure} [ht]
	   \centering
		\includegraphics[width=0.45 \textwidth, 
        height = 0.18\textheight]{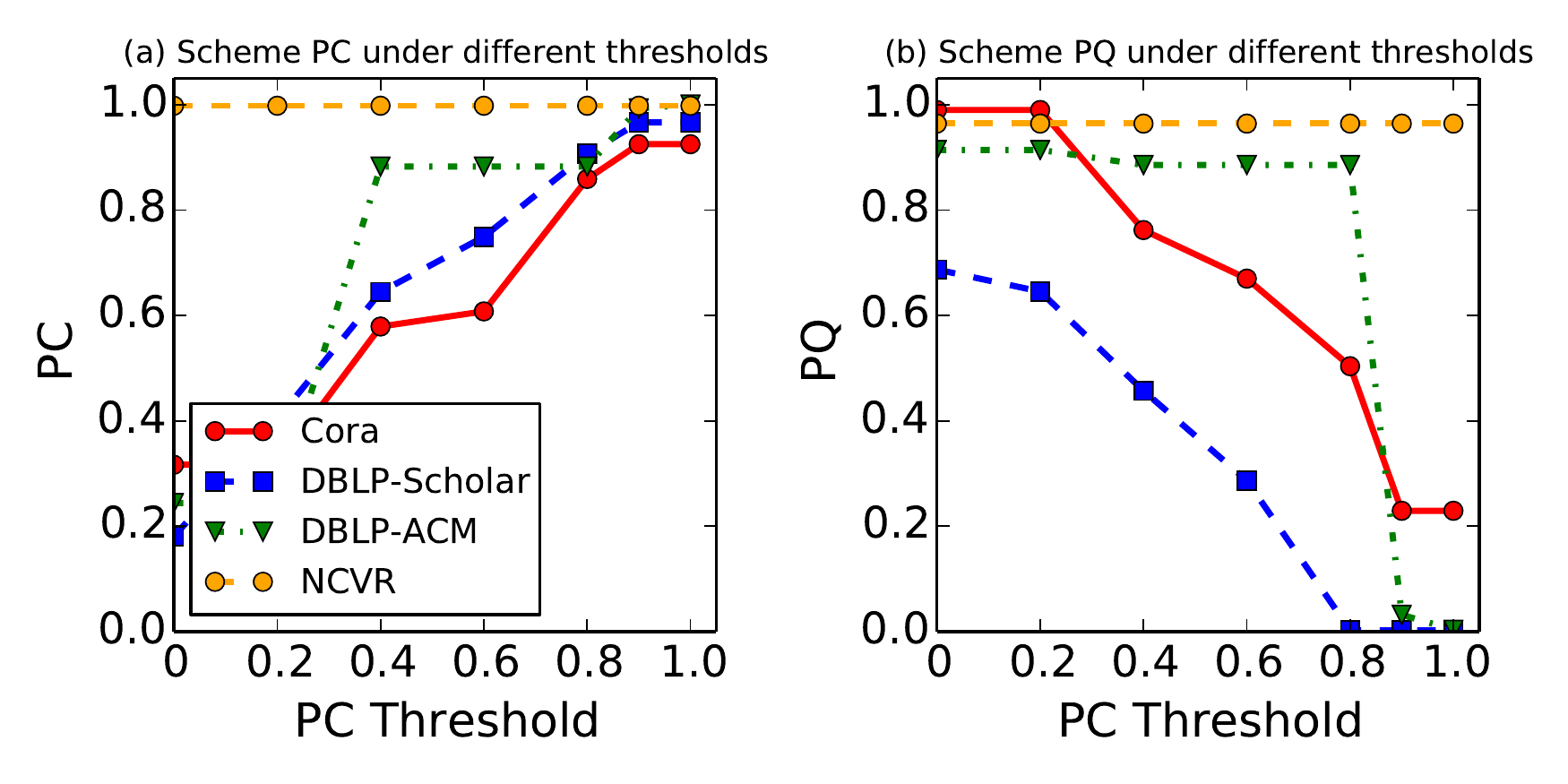}
	\caption{Blocking quality of PC (a) and PQ (b) under different PC thresholds over four datasets} 
	\label{fig:er}
\end{figure}

\subsubsection{Under different aries} The experimental results of the \emph{Pro-Sky} algorithm under different aries are presented in Fig.~\ref{fig:Skyline_ex}. We have also tested on \emph{Naive-Sky} and \emph{Active-Sky} with $\Delta = 0.05$ and $\Delta = 0.10$. However, because the blocking schemes in the scheme skylines learned by these algorithms are the subsets of the scheme skyline generated by \emph{Pro-Sky} in which no $\Delta$ is defined, we thus omit the results for \emph{Naive-Sky} and \emph{Active-Sky}. Fig. \ref{fig:Skyline_ex} presents the progressive process for generating the scheme skylines as well as the blocking schemes from 1-ary to 3+-ary over four datasets. We do not present the further process of 3+-ary, as the scheme skylines have already been generated within 3-ary. In Fig. \ref{fig:Skyline_ex}.(h) and \ref{fig:Skyline_ex}.(l) for NCVR dataset, we present detailed schemes with $PC \in [0.995, 1]$ because most of scheme skylines are located in this range. Such schemes can not be learned by \emph{Naive-Sky} nor \emph{Active-Sky} unless $\Delta$ is set to be very small (e.g. 0.0001). 

\begin{figure*}	[h]
	\begin{center}
		\includegraphics[width= \textwidth]{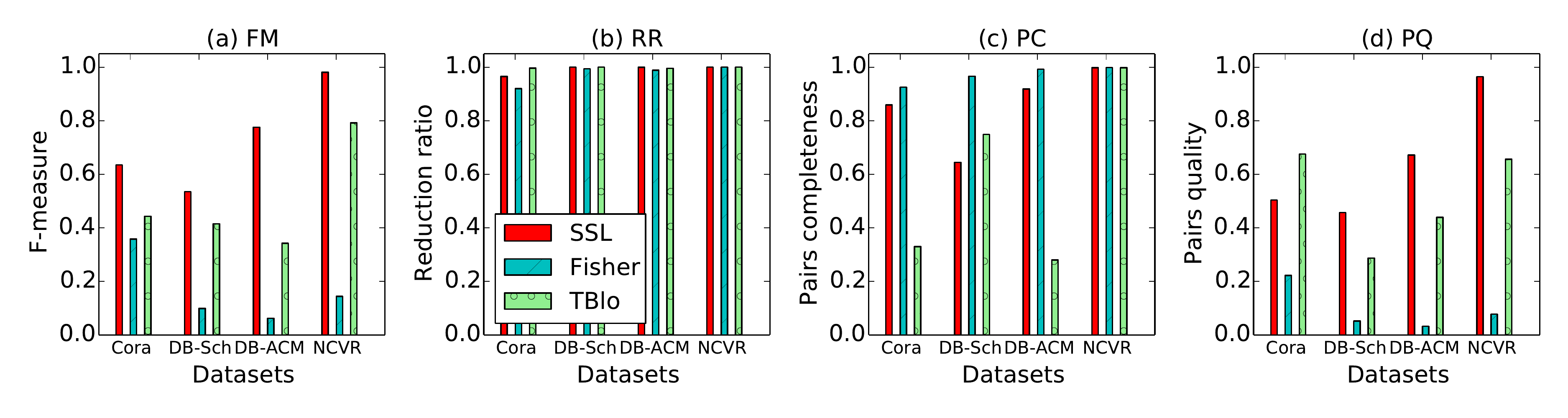}
	\end{center}
\vspace{-10mm}
	\caption{Comparison on blocking quality by different blocking approaches over four datasets using the measures: (a) FM, (b) RR, (c) PC, and (d) PQ}
	\label{fig:performance}
\end{figure*}

\subsubsection{Under different PC thresholds}
To make a detailed comparison with the existing work, we have conducted experiments based on our ASL algorithm under different PC thresholds. We have monitored the blocking schemes learned by the ASL algorithm as well as their PC and PQ values under different PC thresholds $\epsilon \in \{ 0.1, 0.2, 0.4, 0.6, 0.8, 1.0\}$, and the results are presented in Fig. ~\ref{fig:er}. It shows how the PC and PQ vary due to the changing of learned schemes with the increment of the PC threshold with sufficient label budget.


With the increment of threshold, the blocking scheme we learn generates lower PC and higher PQ. However, there are still some points with the same PC and PQ even though the threshold increases. This indicates that under certain thresholds, the algorithm learns the same blocking scheme, which proves, on another side, the efficiency of Active-Sky compared with Naive-Sky. For example, the performance of our approach for NCVR dataset is consistent whatever the PC threshold is, because the blocking scheme learned can generate blocks with both high PC and PQ. On the contrary, for other datasets such as Cora, a lower PC threshold allows the algorithm to seek for blocking schemes that can generate higher PQ, but the PC decreases. We can also notice that for DBLP-ACM and NCVR datasets, blocking schemes with both high PC and PQ are learned with a low threshold (e.g. $\epsilon = 0.4$), but for DBLP-Scholar, no blocking schemes can be learned with high PQ (i.e. higher than 0.6). In the figure, blocking schemes with PC threshold 1.0 are normally hard to achieve, hence we present the maximum PC our scheme can achieve.

\subsubsection{Compared with baselines}

Existing work largely focuses on learning a single blocking scheme, while our approach aims to learn a scheme skyline, which is a set of blocking schemes. Hence we first conduct experiments to present our scheme skyline learning results and show that, in a 2-dimension space of PC and PQ, the scheme points learned by the baselines \emph{TBlo} and \emph{Fisher} are dominated or contained in our scheme skylines over four datasets, as shown in Fig.~\ref{fig:Skyline_ex}(i)-(l).

\begin{table} [h]

	\caption{Comparison on blocking quality where the baseline PC values are selected as thresholds}	
	\centering
	\label{tab:bq}
    
    \begin{tabular}{|c|c|c|c|c|}
    \toprule
    \multirow{2}{*}{} & \multicolumn{2}{|c|}{PC} & \multicolumn{2}{|c|}{PQ} \\
    \cline{2-5}
    & TBlo & ASL & TBlo & ASL \\
    \hline
    Cora & 0.3296 & 0.3167 & 0.6758 & 0.9898 \\
    \hline
    DBLP-Scholar & 0.7492 & 0.7492 & 0.2869 & 0.2869 \\
    \hline
    DBLP-ACM & 0.2801 & 0.8826 & 0.4387 & 0.8854 \\
    \hline
    NCVR & 0.9981 & 0.9979 & 0.6558 & 0.9640 \\    
    \hline
    \end{tabular}
    
    \vspace{5mm}
    
    \begin{tabular}{|c|c|c|c|c|}
    \toprule
    \multirow{2}{*}{} & \multicolumn{2}{|c|}{PC} & \multicolumn{2}{|c|}{PQ} \\
    \cline{2-5}
    & Fisher & ASL & Fisher & ASL \\
    \hline
    Cora & 0.9249 & 0.9249 & 0.2219 & 0.2219 \\
    \hline
    DBLP-Scholar & 0.9928 & 0.9928 & 0.0320 & 0.0320 \\
    \hline
    DBLP-ACM & 0.9661 & 0.9686 & 0.0522 & 0.6714 \\
    \hline
    NCVR & 0.9990 & 0.9990 & 0.0774 & 0.0774 \\    
    \hline
    \end{tabular}
	
\end{table}

To make a fair point-to-point comparison with baselines, we conduct experiments which regard the baseline PC values as PC thresholds for the ASL algorithm, and compare the PC and PQ values which are listed in Table~\ref{tab:bq}. Comparing with \emph{TBlo}, our approach can generate blocks with much higher (i.e. from 50\% to 101\%) PQ while remaining similar PC except in \emph{DBLP-Scholar}, where the results are the same. Comparing with \emph{Fisher}, in the most case, the results are the same, except in \emph{DBLP-Scholar}, the results generated by our approach have a 12 times higher PQ. In general, our approach can generate results as good as or better than the baseline approaches under same thresholds with sufficient label budget.

We also present the schemes with highest f-measure values in the scheme skylines and compare them with baseline schemes in terms of FM, RR, PC and PQ. The FM results are shown in Fig.~\ref{fig:performance}(a) in which our approach outperforms all the baselines over all the datasets. In Fig.~\ref{fig:performance}(b), all the approaches yield high RR values over four datasets. In Fig.~\ref{fig:performance}(c), the PC values of our approach are not the highest over the four datasets, but they are not much lower than the highest one (i.e. within 10\% lower except in DBLP-Scholar). However, out approach can generate higher PQ values than all the other approaches, from 15 percents higher in NCVR (0.9956 vs 0.8655) to 20 times higher in DBLP-ACM (0.6714 vs 0.0320), as shown in Fig.~\ref{fig:performance}(d).

%% file: Related_work.tex
\section{Related Work}
\label{sec_rw}


Scheme based blocking techniques for entity resolution were first mentioned by Fellegi and Sunter 
\cite{fellegi1969theory}.
They used a pair \emph{(attribute, matching-method)} to define blocks. 
For example, the soundex code of both names \emph{Gail} and \emph{Gayle} is ``g400'', thus records that contain either of the names will be placed into the same block. At that time, a blocking scheme was chosen by domain experts without using any learning algorithm. After that, scheme learning approaches generally fall into two categories: (1) supervised blocking scheme learning approaches \cite{michelson2006learning,cao2011leveraging}, and (2) Unsupervised blocking scheme learning approaches \cite{kejriwal2013unsupervised, kejriwal2014two, kejriwal2015dnf}. 

Michelson and Knoblock \cite{michelson2006learning} first proposed a blocking scheme learning algorithm, called \emph{Blocking Scheme Learner}, which is the first supervised algorithm to learn blocking schemes.
This approach adopts the Sequential Covering Algorithm(SCA) \cite{mitchell1997machine} to learn schemes with both high PC and high RR. In the same year, Bilenko et al. \cite{bilenko2006adaptive} proposed two blocking scheme learning algorithms called \emph{ApproxRBSetCover} and \emph{ApproxDNF} to learn disjunctive blocking schemes and DNF (i.e. Disjunctive Normal Form) blocking schemes, respectively. Both algorithms are supervised and need training samples with labels.
Then, Cao et al.  \cite{cao2011leveraging} noticed that obtaining labels for training samples is very expensive, and thus used both labeled and unlabeled samples in their approach. Their algorithm can learn a blocking scheme using conjunctions of blocking predicates which satisfy both minimum true-match coverage and minimum precision criteria. 

Later on, Kejriwal et al. \cite{kejriwal2013unsupervised} proposed an unsupervised algorithm for learning blocking schemes, where no labels from human labelers are needed. Instead, a weak training set was applied, where both positive and negative labels were generated
by calculating the similarity of two records in terms of TF-IDF. 
The predicate with the highest score is selected as part of the result, and if the lower ranking predicate can cover more positively labeled pairs in the training set, they will be selected in a disjunctive form. After traversing all the predicates, a blocking scheme is learned. Although this approach circumvents the need of human labelers, using unlabeled samples or labeled samples only based on string similarity are not reliable \cite{wang2016semantic} and hence blocking quality can not be guaranteed \cite{wang2016clustering}. 
We need to notice that the Meta-blocking approaches proposed by \cite{papadakis2014meta,papadakis2014supervised} are proposed based on existing blocking results, and is not targeting at learning blocking schemes.



However, so far, all related approaches have focused on learning a single blocking scheme under given constraints. No work has been reported to provide an overview of all possible blocking schemes under different constraints for the users so that they can choose their preferred blocking schemes under their own constraints. To fill in this gap, in this paper, we consider to use skyline techniques to present a set of blocking schemes under different constraints.

The concept of \emph{skyline query} has been widely studied in the context of databases. A good number of approaches have been developed in recent years \cite{chomicki2013skyline,lin2007selecting,tao2009distance,sarma2011representative}, which primarily focused on learning representative skylines, such as top-k RSP (representative skyline points) \cite{lin2007selecting}, k-center (i.e. choose $k$ centers and one skyline point for each center) \cite{tao2009distance} and threshold-based preference \cite{sarma2011representative}. A survey by Kalyvas and Tzouramanis \cite{kalyvas2017survey} has reviewed the major approaches in this area.

From an algorithmic perspective, a naive algorithm for skyline queries (e.g., nested-loop algorithm) has the time complexity $O(n^2d)$, where $n$ is the number of records and $d$ is the number of attributes in a given database. Later, several algorithms have been proposed to improve the efficiency of skyline queries based on different properties, which have been previously ignored by the naive algorithm. In the early days, Borzsony et al. \cite{borzsony2001skyline} proposed the BNL (block nested-loop) algorithm based on the transitivity of the dominance relation (e.g. if $a$ dominates $b$ and $b$ dominates $c$, then $a$ dominates $c$) . Then, Chomicki et al. \cite{chomicki2003skyline,chomicki2005skyline} proposed an SFS (sort-filter-skyline) algorithm with the improvements: progressive and optimal comparison times. Sheng and Tao proposed an EM (external memory) model based on the attribute order, which was discussed by Borzsonyi et al. \cite{borzsony2001skyline}. Morse et al. proposed the LS-B (lattice skyline) algorithm based on the low cardinality of some attributes (e.g. the rating of movies are integers within a small range of [1, 5]) \cite{morse2007efficient}. Papadias et al. proposed the BBS (branch-and-bound skyline) algorithm based on the index of all input records by an R-tree \cite{papadias2003optimal}. 

Nevertheless, existing work on skyline queries aimed to efficiently tease out the skyline of queries over a database in which records and attributes are known. In contrast, our study in this paper has shifted the focus to learning the skyline of blocking schemes in terms of a given number of selection criteria but the actual values of those selection criteria are not available in a database. Particularly, in many real-world applications, only a limited number of labels are allowed to be used for assessing blocking schemes. Thus, how to efficiently and effectively learning the skyline of blocking schemes is a difficult task, as previously discussed. To overcome the difficulty of limited label budgets, in this paper, we consider to leverage active learning techniques for finding informative samples and improving the performance of learning.

Active learning has been extensively studied in the past \cite{settles2012active}. Ertekin et al. \cite{ertekin2007learning} showed that active learning may provide almost the same or even better results in solving the class imbalance problem, compared with the random sampling approaches, such as oversampling the minority class and/or undersampling the majority class \cite{chawla2002smote}. In recent years, a number of active learning approaches have been introduced for entity resolution \cite{arasu2010active,bellare2012active,fisher2016active}. For example, Arasu et al. \cite{arasu2010active} proposed an active learning algorithm based on the monotonicity assumption, i.e. the more textually similar a pair of records is, the more likely it is a matched pair. Their algorithm aimed to maximize recall under a manually specific precision constraint.
To reduce the label and computational complexity, Bellare et al. \cite{bellare2012active} proposed an approach to solve the maximal recall with precision constraint problem by converting this problem into a classifier learning problem. 
In their approach, the main algorithm is called \emph{ConvexHull} which aimed to find the best classifier with the maximal recall by finding the classifier with the minimal 01-Loss. Here the 01-Loss refers to the total number of false negatives and false positives. The authors then designed another algorithm called \emph{RejectionSampling} which used a black-box to generate the 01-Loss of each classifier. The black-box invokes the \emph{IWAL} algorithm in \cite{beygelzimer2010agnostic}.

In this paper, for the first time, we study the problem of scheme skyline learning. Previous work on blocking scheme learning aims to learn a single scheme that matches the user requirements. Although skyline technique can help to provide a set of schemes but traditionally, they regard all the attribute values (e.g. PC and PQ) as known, which do not apply to scheme skyline. 
We thus aim to propose novel algorithms to efficiently learn a scheme skyline. 




%% file: Conclusions.tex
\section{Conclusions}\label{sec_con}

In this paper, we have proposed the scheme skyline learning approach called skyblocking, which uses skyline query techniques and active learning techniques to learn a set of optimal blocking schemes under different constraints and a limited label budget. 
We have tackled the class imbalance problem by solving the balanced sampling problem which is proved to be more label efficient than random sampling. 
We have also proposed the scheme extension strategy to reduce the searching space and the label cost. Three algorithms are proposed for efficiently learning scheme skylines.
Additionally, our approach overcomes the weaknesses of existing blocking scheme learning approaches in that: (1) Previously, supervised blocking scheme learning approaches require a large number of labels for learning a blocking scheme, which is an expensive task for entity resolution; (2) Existing unsupervised approaches generate training sets based on the similarity of record pairs, instead of true labels, thus the training quality can not be guaranteed.




%% file: Acknowledgment.tex
\section*{Acknowledgment}

  This work was partially funded by the Australian Research Council (ARC) under Discovery Project DP160101934.